\documentclass[submission,copyright,creativecommons]{eptcs}
\providecommand{\event}{CL\&C 2010}
\providecommand{\urlalt}[2]{\href{#1}{#2}} \providecommand{\doi}[1]{doi:\urlalt{http://dx.doi.org/#1}{#1}} 

\usepackage{includes}

\title{A \cbv{} \lambdacalculus{} with lists and control}
\author{Robbert Krebbers
	\institute{Radboud University Nijmegen}
	\email{mail@robbertkrebbers.nl}
}

\newcommand{\titlerunning}{{A \cbv{} lambda-calculus with lists and control}}
\newcommand{\authorrunning}{{Robbert Krebbers}}

\begin{document}
\maketitle{}

\begin{abstract}
Calculi with control operators have been studied to reason about control in programming languages and to interpret the computational content of classical proofs.
To make these calculi into a real programming language, one should also include \datatype{}s.

As a step into that direction, this paper defines a simply typed \cbv{} \lambdacalculus{} with the control operators \(\catch\) and \(\throw\), a \datatype{} of lists, and an operator for primitive recursion (\`a la \GoedelsTfull).
We prove that our system satisfies subject reduction, progress, confluence for untyped terms, and strong normalization for well-typed terms.
\end{abstract}

\section{Introduction}
The extension of simply typed \lambdacalculus{} with control operators and the observation that these operators can be typed using rules of classical logic is originally due to Griffin~\cite{griffin1990} and has lead to a lot of research by varying the control operators, the underlying calculus or the computation rules, or by studying concrete examples of the computational content of classical proofs.
Little of this research has considered the problem of how to incorporate primitive \datatype{}s in direct style.
If one wants to use these calculi as a real functional programming language with control, this is a gap that needs filling.

This paper contributes towards the development of a \lambdacalculus{} with both \datatype{}s and control operators that allows program extraction from classical proofs.
In such a calculus one can write specifications of programs, which can be proven using (a restricted form of) classical logic. 
Program extraction would then allow to extract a program from such a proof where the classical reasoning steps are extracted to control operators.
This approach yields programs-with-control that are \emph{correct by construction} because they are extracted from a proof of the specification. 
However, in order for these extracted programs to be useful in practice, \datatype{}s in direct style should be supported.

As a step into that direction, we introduce \lambdaCatch, a simply typed \cbv{} \lambdacalculus{} with the control operators \catch{} and \throw, a list and unit \datatype{}, and an operator for primitive recursion (\`a la \GoedelsTfull).
We consider lists because those are among the most commonly used \datatype{}s in functional programming.
Expressively, lists make our system as least as strong as \GoedelsTfull{} because natural numbers can be encoded as lists over the unit type.
We prove the conventional meta theoretical properties -- subject reduction, progress, confluence, and strong normalization -- so that it may be used as a sound basis for a calculus that allows program extraction from classical proofs.

Our system is based on Herbelin's \IQCMP{}-calculus with \catch{} and \throw{} that he uses to give a computational interpretation of Markov's principle~\cite{herbelin2010}.
Most importantly, we adopt his restriction of the control operator \catch{} to \arrowFree{} types.
This restriction enables the system to satisfy \emph{progress} without losing other meta theoretical properties.
The progress property states that if \(t\) is a well-typed closed term, then \(t\) is either a value or there is a term \(t'\) such that \(t\) reduces to \(t'\). 
From a programmer's point of view this is an important property as together with confluence it ensures \emph{unique representation of data}.
For example, for the natural numbers, unique representation of data means that for each natural number there is (up to conversion) a unique closed term of the type of natural numbers.
To show how the system can be used in programming, we give a simple example in~\ref{example:program}, where we define a function that multiplies the values of a list and throws an exception as soon as it encounters the value \(0\).

Proving confluence or strong normalization for systems with control generally requires complex extensions of standard proof methods, see for example~\cite{parigot1997,py1998,baba2001,nakazawa2003,geuvers2011,rehof1994}.
For \lambdaCatch{} this is less the case.
We give relatively short proofs of subject reduction, progress, confluence for untyped terms, and strong normalization for well-typed terms.

\subsection{Related work}
Incorporating \datatype{}s into a \lambdacalculus{} with control has not received much attention. 
We briefly summarize the research done in this direction and compare it with our work. 

Parigot~\cite{parigot1992} has described a variant of his \lambdaMu{}-calculus with second-order types.
His system is very powerful, because all the well-known second-order representable \datatype{}s are included in it.
But as observed in~\cite{parigot1992,parigot1993}, it does not ensure unique representation of data.
This defect can be remedied by adding additional reduction rules, however, this results in a loss of confluence.
Another approach is to use output operators to extract data, but this introduces an additional indirection.

Rehof and S{\o}rensen have described an extension of their \(\lambda_\Delta\)-calculus with basic constants and functions~\cite{rehof1994}.
Unfortunately their extension is quite limited. 
In particular, an operator for primitive recursion, which takes terms rather than basic constants as its arguments, cannot be defined.

Barthe and Uustalu~\cite{barthe2002} have considered CPS-translations for inductive and coinductive types.
In particular, they describe a system with a primitive for iteration over the natural numbers, and the control operator \(\Delta\).
They prove preservation of typing and reduction under a CPS-translation, but do not consider other meta theoretical properties of this system.

Crolard and Polonowski~\cite{crolard2011} have considered a version of \GoedelsTfull{} with products and \callcc.
However, as their semantics is presented by CPS-translations instead of a direct specification via a calculus, their work is not directly related to ours.

Geuvers, Krebbers and McKinna~\cite{geuvers2011} have defined an extension of Parigot's \lambdaMu-calculus with a \datatype{} of natural numbers and an operator for primitive recursion.
They prove that their system satisfies subject reduction, unique representation of the naturals, confluence and strong normalization.
Also, they define a CPS-translation into \GoedelsTfull{} to show that adding control operators does not extend the expressive power.
Unfortunately, their system is call-by-name with call-by-value evaluation for \datatype{}s, making it less suitable to model control in most programming languages.
Due to their decision to use \lambdaMu{}, their proofs involve many complex extensions of standard proof techniques, and expose a lot of non-trivial interaction between control and \datatype{}s.

Several extensions of \lambdacalculus{} with the control operators \(\catch\) and \(\throw\) have been studied in the literature.
We discuss those that are most relevant to our work.
Crolard~\cite{crolard1999} has considered a \cbn{} variant of such a calculus, for which he defines a correspondence with Parigot's \lambdaMu{}-calculus.
He uses this correspondence to prove confluence, subject reduction and strong normalization, but does not consider \datatype{}s in direct style.

Herbelin~\cite{herbelin2010} has defined \IQCMP{}, a calculus with \(\catch\) and \(\throw\) to give a computational interpretation of Markov's principle.
His calculus is \cbv{} and supports product, sum, existential, and universally quantified types.
An essential feature of his calculus is the restriction of \(\catch\) to \(\forall\)-\(\arrow\)-free types.
This restriction enables him to prove progress, which is an important property for his main result, a proof of the disjunction and existence property.

Since Herbelin's \IQCMP{}-calculus has a convenient meta theory, we use it as the starting point for our work.
But instead of considering product, sum, existential, and universally quantified types, we consider a data type of lists in direct style.
Whereas Herbelin does not consider confluence, and does not give a direct proof of strong normalization, we will give direct proofs of these properties for our system.

\subsection{Outline}
In Section~\ref{section:system}, we define the typing rules, and the basic reduction rules, whose compatible closure defines computation in \lambdaCatch.
We give two example programs showing interaction between \datatype{}s and control.
Section~\ref{section:system} moreover contains proofs of subject reduction and progress.
Section~\ref{section:confluence} contains a direct proof of confluence for untyped terms based on an analysis of complete developments.
Section~\ref{section:sn} contains a direct proof of strong normalization using the reducibility method.
We close with conclusions and indications for further work in Section~\ref{section:conclusions}.

\section{The system}
\label{section:system}

\begin{definition}
\label{definition:system}
The \emph{types}, \emph{terms} and \emph{values} of \(\lambdaCatch\) are defined as
\begin{flalign*}
\sigma,\tau,\rho \inductive{}&
	\unitTypeT \separator \listTypeT \tau \separator \sigma \arrow \tau \\
t,r,s \inductive{}&
	x
	\separator \unitT 
	\separator \nilT 
	\separator (\consTsymbol)
	\separator \lrecTsymbol
	\separator \lambda x.r 
	\separator ts 
	\separator \catchin \alpha t 
	\separator \throwto \alpha t \\
v,w,v_r,v_s \inductive{}&
	x
	\separator \unitT
	\separator \nilT
	\separator (\consTsymbol)
	\separator (\consTsymbol)\;v
	\separator (\consTsymbol)\;v\;w
	\separator \lrecTsymbol
	\separator \lrecTsymbol\;v_r
	\separator \lrecTsymbol\;v_r\;v_s
	\separator \lambda x.r
\end{flalign*}
where \(x\), \(y\), and \(z\) range over \emph{variables}, and \(\alpha\), \(\beta\) and \(\gamma\) range over \emph{continuation variables}.
\end{definition}

The construct \(\lambda x.r\) binds \(x\) in \(r\), and \(\catchin \alpha t\) binds \(\alpha\) in \(t\).
The precedence of \(\lambda\) and \(\catch\) is lower than application, so instead of \(\catchin \alpha (tr)\) we write \(\catchin \alpha tr\).
We let \(\FV t\) denote the set of free variables of \(t\), and \(\FCV t\) the set of free continuation variables of \(t\). 
As usual, we use \emph{Barendregt's variable convention}~\cite{barendregt1984}. 
That is, given a term, we may assume that bound variables are distinct from free variables and that all bound variables are distinct. 
The operation of capture avoiding substitution \(\subst t x r\) of \(r\) for \(x\) in \(t\) is defined in the usual way.

The constructs \(\nilT\) and \((\consTsymbol)\) are the constructors of the list \datatype{}.
We treat these constructors, and the operator \(\lrecTsymbol\) for primitive recursion over lists, as unary constants so we can use them in partially applied position.
Also, this treatment results in a more uniform definition of the reduction rules.
We often use \Haskell{}-style notation.
In particular, we write \(\consT t r\) to denote \((\consTsymbol)\;t\;r\), and \(\lambda \_\,.\,t\) to denote \(\lambda x.t\) with \(x \notin \FV t\).
Furthermore, we write \(\listenc{t_1,\ldots,t_n}\) to denote \(\consT {t_1} {\consT {\ \ldots\ } {\consT {t_n} \nilT}}\).

Following Herbelin~\cite{herbelin2010} we restrict \catch{} to \arrowFree{} types.
Without this restriction, progress (Theorem~\ref{theorem:progress}) would fail. 
Let us consider the term \(\catchin \alpha {\lambda x.\throwto \alpha {(\lambda y.y)}}\).
Without this restriction, this term would have had type \(\unitTypeT \arrow \unitTypeT\), whereas it would not reduce to a value.
In fact, even \((\catchin \alpha {\lambda x.\throwto \alpha {(\lambda y.y)})}\ \unitT : \unitTypeT\) would not reduce.
The reduction rules for \catch{} and \throw{} are very similar to~\cite{herbelin2010}, but quite different from those by Crolard~\cite{crolard1999}.
In particular, Crolard includes reduction rules to move the \catch{} whereas Herbelin's system and ours merely allow a \throw{} to move towards the corresponding \catch{}.
This is due to the restriction to \arrowFree{} types.

\begin{definition}
We let \(\phi\) and \(\psi\) range over \arrowFree{} types.
\end{definition}

\begin{definition}
\label{definition:typing}
Let \(\Gamma\) be a map from variables to types, and let \(\Delta\) be a map from continuation variables to \arrowFree{} types.
The derivation rules for the typing judgment \(\typed \Gamma \Delta t \rho\) are as shown below.
\begin{center}

\vspace{0.1cm}

	\AXC{\strut\(x : \rho \in \Gamma\)}
	\UIC{\(\typed \Gamma \Delta x \rho\)}
	\normalAlignProof
	\DisplayProof\qquad
	\AXC{\strut}
	\UIC{\(\typed \Gamma \Delta \unitT \unitTypeT\)}
	\normalAlignProof
	\DisplayProof \qquad
	\AXC{\strut}
	\UIC{\(\typed \Gamma \Delta \nilT {\listTypeT\sigma}\)}
	\normalAlignProof
	\DisplayProof \qquad
	\AXC{\strut}
	\UIC{\(\typed \Gamma \Delta {(\consTsymbol)} {\sigma \arrow \listTypeT\sigma \arrow \listTypeT\sigma}\)}
	\normalAlignProof
	\DisplayProof

\vspace{0.5cm}
	\AXC{}
	\UIC{\(\typed \Gamma \Delta \lrecTsymbol {\rho \arrow (\sigma \arrow \listTypeT\sigma \arrow \rho \arrow \rho) \arrow \listTypeT\sigma \arrow \rho}\)}
	\normalAlignProof
	\DisplayProof
	
\vspace{0.5cm}
	\AXC{\(\typed {\Gamma, x : \sigma} \Delta t \tau\)}
	\UIC{\(\typed \Gamma \Delta {\lambda x.t} {\sigma \arrow \tau}\)}
	\normalAlignProof
	\DisplayProof\qquad
	\AXC{\(\typed \Gamma \Delta t {\sigma \arrow \tau}\)}
	\AXC{\(\typed \Gamma \Delta s \sigma\)}
	\BIC{\(\typed \Gamma \Delta {ts} \tau\)}
	\normalAlignProof
	\DisplayProof

\vspace{0.5cm}
	\AXC{\(\typed \Gamma {\Delta,\alpha:\psi} t \psi\)}
	\UIC{\(\typed \Gamma \Delta {\catchin \alpha t} \psi\)}
	\normalAlignProof
	\DisplayProof \qquad
	\AXC{\(\typed \Gamma \Delta t \psi\)}
	\AXC{\(\alpha:\psi \in \Delta\)}
	\BIC{\(\typed \Gamma \Delta {\throwto \alpha t} \tau\)}
	\normalAlignProof
	\DisplayProof
\vspace{0.4cm}
\end{center}
\end{definition}

\begin{lemma}
\label{lemma:values_typed}
Given a value \(v\) with \(\typed {} \Delta v \rho\), then:
\begin{enumerate}
\item If \(\rho = \unitTypeT\), then \(v\) is of the shape \(\unitT\).
\item If \(\rho = \listTypeT \sigma\), then \(v\) is of the shape \(\listenc{w_1, \ldots, w_n}\).
\item \label{lemma:values_typed_arrow}
	If \(\rho = \sigma\arrow\tau\), then \(v\) is of the shape \((\consTsymbol)\), \((\consTsymbol)\,w\), \(\lrecTsymbol\), \(\lrecTsymbol\,v_r\), \(\lrecTsymbol\,v_r\,v_s\) or \(\lambda x.r\).
\end{enumerate}
\end{lemma}

\begin{proof}
This result is proven by induction on the structure of \(v\). 
The case \(v \equiv x\) is impossible because \(v\) is closed for free variables. 
The other cases are easy.
\end{proof}

\begin{definition}
\label{definition:context}
The \emph{contexts} of \(\lambdaCatch\) are defined as:
\[
	E \inductive \Box t \separator v \Box \separator \throwto \alpha \Box
\]
Given a context \(E\) and a term \(s\), the \emph{substitution of \(s\) for the hole in \(E\)}, notation \(\cctx E s\), is defined in the usual way.
\end{definition}

\begin{definition}
\label{definition:reduction}
Reduction \(t \red t'\) is defined as the compatible closure of:
\begin{flalign*}
	(\lambda x.t)\;v \red{}& \subst t x v \tag{\(\betav\)} \\
	\cctx E {\throwto \alpha t} \red{}& \throwto \alpha t \tag{\(\throwRule\)} \\
	\catchin \alpha {\throwto \alpha t} \red{}& \catchin \alpha t \tag{\(\catchRule{1}\)} \\
	\catchin \alpha {\throwto \beta v} \red{}& \throwto \beta v \ \textnormal{\small if \(\alpha \notin \{\beta\}\cup\FCV v\)} \tag{\(\catchRule{2}\)} \\
	\catchin \alpha v \red{}& v \hspace{1.30cm} \textnormal{\small if \(\alpha \notin\FCV v\)} \tag{\(\catchRule{3}\)} \\
	\lrecT {v_r} {v_s} \nilT \red{}& {v_r} \tag{\(\nilT\)} \\
	\lrecT {v_r} {v_s} {(\consT {v_h} {v_t})} \red{}& v_s\;v_h\;v_t\;(\lrecT {v_r} {v_s} {v_t}) \tag{\(\consTsymbol\)}
\end{flalign*}
As usual, \(\redd\) denotes the reflexive/transitive closure and \(\conv\) denotes the reflexive/symmetric/transitive closure.
\end{definition}

Notice that because we treat partially applied \((\consTsymbol)\) and \(\lrecTsymbol\) constructs as values, we get reductions like \(\consT {\throwto \alpha r} t \equiv (\consTsymbol) \;(\throwto \alpha r)\; t \red (\throwto \alpha r)\; t \red \throwto \alpha r\) for free without the need for additional contexts for \((\consTsymbol)\) and \(\lrecTsymbol\).

\begin{fact}
\label{fact:value_closed_vars}
If \(\typed \Gamma \Delta v \psi\), then \(\FCV v = \emptyset\)
\end{fact}

\begin{proof}
By induction on the structure of the value \(v\).
Since \(\psi\) is \(\arrow\)-free, we only have to consider the cases \(v \equiv x\), \(v \equiv \unitT\), \(v \equiv \nilT\) and \(v \equiv \consT {v_l} {v_r}\), for which the result trivially holds.
\end{proof}

The reduction rules (\(\catchRule 2\)) and (\(\catchRule 3\)) require that \(\alpha \notin \FCV v\).
This side condition can be omitted for well-typed terms by the previous fact.
However, since we consider the problem of confluence for untyped terms (Section~\ref{section:confluence}), we do need this additional restriction. 

\begin{definition}
\label{definition:naturals}
We define a type for the natural numbers \(\natTypeT \defined \listTypeT \unitTypeT\), with the following operations on it.
\begin{flalign*}
	\zeroT \defined{}& \nilT \\
	\sucTsymbol \defined{}& (\consTsymbol)\; \unitT \\
	\nrecTsymbol \defined{}& \lambda x_r x_s\,.\,\lrecTsymbol\; x_r\; (\lambda \,\_\,.\,x_s)
\end{flalign*}
We let \(\natenc n \defined \sucTsymbol^n \zeroT\) denote the representation of a natural number.
\end{definition}

\begin{fact}
The operations on \(\natTypeT\) satisfy the expected conversions.
\begin{flalign*}
	\nrecT {v_r} {v_s} \zeroT \redd{}& v_r  \\
	\nrecT {v_r} {v_s} {(\sucT v)} ={}& v_s\;v\;(\nrecT {v_r} {v_s} v)
\end{flalign*}
\end{fact}

Colson and Fredholm~\cite{colson1998} have shown that in \GoedelsTfull{} with call-by-value reduction, it takes at least a number of steps that is linear with respect to the input for a non-trivial algorithm to reduce to a value.
In particular, it is impossible to compute the predecessor in constant time.
Intuitively it is easy to see why, consider the reduction \(\nrecT {v_r} {v_s} {(\sucT v)} \red v_s\;v\;(\nrecT {v_r} {v_s} v)\).
Due to the restriction of \(\beta\)-reduction to values, the recursive call, \(\nrecT {v_r} {v_s} v\) has to be reduced to a value before the whole term is able to reduce to a value.
In \lambdaCatch{} we can use the control mechanism to do better.

\begin{example}
\label{example:predecessor}
We define the predecessor function \(\predT : \natTypeT \arrow \natTypeT\) as follows.
\[
	\predT \defined \lambda n\,.\,\catchin \alpha {\nrecT \zeroT {(\lambda x\,.\,\throwto \alpha x)} n}
\]
Computing the predecessor is possible in a constant number of steps.
\begin{flalign*}
\predT\ {\natenc{n+1}} 
	\redd{}& \catchin \alpha {\nrecT \zeroT {(\lambda x\,.\,\throwto \alpha x)} {(\sucT{\natenc n})}} \\
	\redd{}& \catchin \alpha {(\lambda x\,.\,\throwto \alpha x)\;\natenc n\;(\lrecT \zeroT {(\lambda \,\_\,x\,.\,\throwto \alpha x)} {\natenc n})} \\
	\redd{}& \catchin \alpha {(\throwto \alpha {\natenc n})\;(\lrecT \zeroT {(\lambda \,\_\,x\,.\,\throwto \alpha x)} {\natenc n})} \\
	\redd{}& \catchin \alpha {\throwto \alpha {\natenc n}} \redd \natenc n
\end{flalign*}
\end{example}

\begin{example}
\label{example:program}
We define a \lambdaCatch-program \(F : \listTypeT\natTypeT\arrow\natTypeT\) that computes the product of the elements of a list.
The interest of this program is that it uses the control mechanism to stop multiplying once the value 0 is encountered.
\begin{flalign*}
	F \defined{}& \lambda l\, .\, \catchin \alpha {\lrecT {\natenc 1} H l} \\
	H \defined{}& \lambda x\,\_\,.\,\nrecT {(\throwto \alpha \zeroT)} {(\lambda y \,\_\, h\,.\, \sucT y * h)} x
\end{flalign*}
Here, addition \((+)\) and multiplication \((*)\) are defined as follows.
\begin{flalign*}
	(+) \defined{}& \lambda n m \,.\,\nrecT m {(\lambda \,\_\; y \,.\, {\sucT y})} n \\
	(*) \defined{}& \lambda n m \,.\,\nrecT \zeroT {(\lambda \,\_\; y \,.\, m + y)} n
\end{flalign*}
We show a computation of \(F\,\listenc{\natenc 4, \natenc 0, \natenc 9}\).
\begin{flalign*}
F\,\listenc{\natenc 4, \natenc 0, \natenc 9} 
	\redd{}& \catchin \alpha {\lrecT {\natenc 1} H {\listenc{\natenc 4, \natenc 0, \natenc 9}}} \\
	\redd{}& \catchin \alpha {\nrecT {(\throwto \alpha \zeroT)} {(\lambda y \,\_\,h\,.\, \sucT y * h)} {\natenc 4}\;(\lrecT {\natenc 1} H {\listenc{\natenc 0, \natenc 9}})} \\
	\redd{}& \catchin \alpha {(\lambda h\,.\,\natenc 4 * h)\; (\lrecT {\natenc 1} H {\listenc{\natenc 0, \natenc 9}})} \\
	\redd{}& \catchin \alpha {(\lambda h\,.\,\natenc 4 * h)\; (\throwto \alpha \zeroT)} \\
	\redd{}& \catchin \alpha {\throwto \alpha \zeroT}
	\redd{} \zeroT
\end{flalign*}
\end{example}

\begin{lemma}
\label{lemma:typed_subst}
If \(\typed \Gamma \Delta r \sigma\) and \(\typed {\Gamma, x :\sigma} \Delta t \rho\), then \(\typed \Gamma \Delta {\subst t x r} \rho\).
\end{lemma}

\begin{theorem}[Subject reduction]
\label{theorem:sr}
If \(\typed \Gamma \Delta t \rho\) and \(t \red t'\), then \(\typed \Gamma \Delta {t'} \rho\).
\end{theorem}

\begin{proof}
We have to show that each reduction rule preserves typing. 
We use Lemma~\ref{lemma:typed_subst} for (\(\betav\)).
\end{proof}

\begin{lemma}
\label{lemma:nf_open}
Given a normal form \(t\) with \(\typed {} \Delta t \rho\), then either \(t\) is a value, or \(t \equiv \throwto \beta v\) for some value \(v\) and continuation variable \(\beta\).
\end{lemma}

\begin{proof}
This result is proven by induction on the derivation of \(\typed {} {\Delta} t \rho\).
\begin{enumerate}
\item Let \(\typed {} \Delta x \rho\) with \(x : \rho \in \emptyset\). 
	This is impossible because \(x : \rho \notin \emptyset\).
\item In the case of \(\unitT\), \(\nilT\), \((\consTsymbol)\), \(\lrecTsymbol\) and \(\lambda x.r\) the result is immediate.
\item Let \(\typed {} \Delta {ts} \tau\) with \(\typed {} \Delta t {\sigma\arrow\tau}\) and \(\typed {} \Delta s \sigma\). 
	By the \IH{} we know that the terms \(r\) and \(s\) are either a value or a \throw.
	Since \(ts\) is in normal form, it is impossible that either of them is a \throw. 
	Therefore, we may assume that both are values.
	Now, since \(t\) has type \(\sigma\arrow\tau\), we can use Lemma~\ref{lemma:values_typed} to analyze the possible shapes of \(t\).
	\begin{enumerate}
	\item Let \(t \equiv \lrecTsymbol\,v_r\,v_s\).
		By the typing rules we obtain that \(s\) has type \(\listTypeT \rho\) for some \(\rho\).
		So, by Lemma~\ref{lemma:values_typed} we have that \(s\) is a list. 
		However, \(ts\) is in normal form, so this is impossible.
	\item Let \(t \equiv \lambda x.r\). This case is impossible because \(s\) is a value and \(ts\) is in normal form.
	\item In all other cases, the term \(ts\) is a value.
	\end{enumerate}
\item Let \(\typed {} \Delta {\catchin \alpha t} \psi\) with \(\typed {} {\Delta,\alpha:\psi} t \psi\). 
	By the \IH{} we know that \(t\) is a value or a \throw.
	If \(t\) is a value, Fact~\ref{fact:value_closed_vars} gives us that \(\alpha \notin \FCV t\).
	This is impossible since \(\catchin \alpha t\) is in normal form.
	Similarly, it is also impossible that \(t\) is a \throw.
\item Let \(\typed {} \Delta {\throwto \alpha t} \sigma\) with \(\typed {} \Delta t \psi\) and \(\alpha:\psi\in\Delta\). 
	By the \IH{} we know that \(t\) is a value or a \throw.
	If \(t\) is a value, we are done.
	Furthermore, \(t\) cannot be a \throw{} since \(\throwto \alpha t\) is in normal form. 
	\qedhere
\end{enumerate}
\end{proof}

\begin{theorem}[Progress]
\label{theorem:progress}
If \(\typed {} {} t \rho\), then \(t\) is either a value, or there is a term \(t'\) with \(t \red t'\). 
\end{theorem}

\begin{proof}
This result follows immediately from Lemma~\ref{lemma:nf_open}.
\end{proof}

\section{Confluence}
\label{section:confluence}
To prove confluence for untyped terms of \lambdaCatch, we use the notion of \emph{parallel reduction}, as introduced by Tait and Martin-L\"of~\cite{barendregt1984}.
A parallel reduction relation \(\pred\) allows to contract a number of redexes in a term simultaneously so as to make it being preserved under substitution.
If one proves that the parallel reduction \(\pred\) satisfies:
\begin{itemize}
\item The \emph{diamond property}: if \(t_1 \pred t_2\) and \(t_1\pred t_3\), then there exists a \(t_4\) such that \(t_2 \pred t_4\) and \(t_3 \pred t_4\).
\item \(t_1 \pred t_2\) implies \(t_1 \redd t_2\) and \(t_1 \redd t_2\) implies \(t_1 \predd t_2\).
\end{itemize}
then one obtains confluence of \(\red\).

Following Takahashi~\cite{takahashi1995}, we further streamline the proof by defining the \emph{complete development} of a term \(t\), notation \(\MPRED t\), which is obtained by contracting all redexes in \(t\).
Now to prove the diamond property of \(\pred\), it suffices to prove that \(t_1 \pred t_2\) implies \(t_2 \pred \MPRED{t_1}\). 

For Parigot's \lambdaMu-calculus, it is well known that the naive parallel reduction is not preserved under substitution~\cite{baba2001}.
Instead, a complex parallel reduction that moves subterms located very deeply in a term towards the outside is needed~\cite{baba2001,nakazawa2003,geuvers2011}.
For \lambdaCatch{} we experience another issue. Consider the following rule.
\begin{center}
If \(t \pred t'\), then \(\cctx E {\throwto \alpha t} \pred \throwto \alpha t'\)
\end{center}
If we take \(\throwto {\alpha_1} {(\throwto {\alpha_2} {(\ldots {\throwto {\alpha_n} \unitT} \ldots )})}\) (with \(n \ge 5\)), then we could perform a reduction that contracts all even numbered \throw{}s, and also a reduction that contracts all odd numbered \throw{}s.
Since these two reducts do not converge in a single parallel reduction step, such a parallel reduction would not be confluent.
To repair this issue we use a similar fix as in~\cite{baba2001,nakazawa2003,geuvers2011}: we allow a \throw{} to jump over a \emph{compound context}.

\begin{definition}
\label{definition:compound_context}
\emph{Compound contexts} are defined as:
\[
	\vec E \inductive \Box \separator \vec E t \separator v \vec E \separator \throwto \alpha {\vec E}
\]
Given a compound context \(\vec E\) and a term \(s\), the \emph{substitution of \(s\) for the hole in \(\vec E\)}, notation \(\cctx {\vec E} s\), is defined in the usual way.
\end{definition}

\begin{definition}
\label{definition:pred}
\emph{Parallel reduction \(t \pred t'\)} is inductively defined as:
\begin{enumerate}
\item \(x \pred x\), \(\unitT \pred \unitT\), \(\nilT \pred \nilT\), \((\consTsymbol) \pred (\consTsymbol)\), and \(\nrecTsymbol \pred \nrecTsymbol\).
\item If \(t \pred t'\) and \(r \pred r'\), then \(tr \pred t'r'\) .
\item If \(t \pred t'\), then \(\lambda x.t \pred \lambda x.t'\).
\item If \(t \pred t'\), then \(\catchin \alpha t \pred \catchin \alpha {t'}\).
\item If \(t \pred t'\) and \(v \pred r\), then \((\lambda x.t)\,v \pred \subst {t'} x r\).
\item If \(t \pred t'\), then \(\cctx {\vec E} {\throwto \alpha t} \pred \throwto \alpha t'\).
\item If \(t \pred t'\), then  \(\catchin \alpha {\throwto \alpha t} \pred \catchin \alpha t'\).
\item If \(v \pred t\) and \(\alpha \notin \{\beta\} \cup \FCV v\), then \(\catchin \alpha {\throwto \beta v} \pred \throwto \beta t\).
\item If \(v \pred t\) and \(\alpha\notin\FV v\), then \(\catchin \alpha v \pred t\).
\item If \(v_r \pred r\), then \(\lrecT {v_r} {v_s} {\nilT} \pred r\).
\item If \(v_r \pred r\), \(v_s \pred s\), \(v_h \pred h\) and \(v_t \pred t\), then \(\lrecT {v_r} {v_s} {(\consT {v_h} {v_t})} \pred s\;h\;t\;(\lrecT r s t)\).
\end{enumerate}
\end{definition}

\begin{lemma}
\label{lemma:pred_help}
Parallel reduction satisfies the following properties.
\begin{enumerate}
\item It is reflexive, i.e. \(t \pred t\).
\item The term \(\subst v x w\) is a value.
\item If \(v \pred t\), then \(t\) is a value.
\item If \(t \pred t'\), then \(\FV {t'} \subseteq \FV t\) and \(\FCV {t'} \subseteq \FCV t\).
\item If \(t \pred t'\) and \(v \pred r\), then \(\subst t x v \pred \subst {t'} x r\).
\end{enumerate}
\end{lemma}

\begin{lemma}
\label{lemma:red_pred}
Parallel reduction enjoys the intended behavior. That is:
\begin{enumerate}
\item If \(t \red t'\), then \(t \pred t'\).
\item If \(t \pred t'\), then \(t \redd t'\).
\end{enumerate}
\end{lemma}

\begin{proof}
The first property is proven by induction on the derivation of \(t \red t'\) using that parallel reduction is reflexive and satisfies the substitution property (Lemma~\ref{lemma:pred_help}). 
The second property is proven by induction on the derivation of \(t \pred t'\) using an obvious substitution lemma for \(\redd\).
\end{proof}

\begin{definition}
\label{definition:complete_development}
The \emph{complete development} \(\MPRED t\) is defined as:
\begin{flalign*}
\MPRED {((\lambda x.t)\,v)} \defined{}& \subst {\MPRED t} x {\MPRED v} \\
\MPRED {(\cctx {\vec E} {\throwto \alpha t})} \defined{}& \throwto \alpha {\MPRED t}
	\qquad \textnormal{if \(t \not\equiv \throwto \gamma s\)} \\
\MPRED {(\catchin \alpha {\throwto \alpha t})} \defined{}& \catchin \alpha {\MPRED t} \\
\MPRED {(\catchin \alpha {\throwto \beta v})} \defined{}& \throwto \beta {\MPRED v}
	\qquad \textnormal{if \(\alpha\notin\{\beta\} \cup \FCV v\)} \\
\MPRED {(\catchin \alpha v)} \defined{}& \MPRED v
	\hspace{2.22cm} \textnormal{if \(\alpha\notin\FCV v\)} \\
\MPRED {(\lrecT {v_r} {v_s} \nilT)} \defined{}& \MPRED{v_r} \\
\MPRED {(\lrecT {v_r} {v_s} (\consT {v_h} {v_t}))} \defined{}& \MPRED{v_s}\;\MPRED {v_h}\;\MPRED {v_t}\;(\lrecT {\MPRED{v_r}} {\MPRED{v_s}} {\MPRED {v_t}})
\end{flalign*}
For variables, \(\unitT\), \(\nilT\), \((\consTsymbol)\) and \(\nrecTsymbol\), the complete development is defined as the identity, and it propagates through the other cases that we have omitted.
\end{definition}

We lift the parallel reduction \(\pred\) to compound contexts with the intended behavior that if \(\vec E \pred \vec F\) and \(q \pred q'\), then \(\cctx {\vec E} {\throwto \alpha q} \pred \cctx {\vec F} {\throwto \alpha {q'}}\).

\begin{definition}
\emph{Parallel reduction \(\vec E \predctx \vec F\)} on compound contexts is inductively defined as:
\begin{enumerate}
\item \(\Box \pred \Box\)
\item \(\throwto \alpha \Box \pred \Box\)
\item If \(\vec E \predctx \vec F\) and \(t \pred t'\), then \(\vec E t \predctx \vec F t'\).
\item If \(\vec E \predctx \vec F\) and \(v \pred t\), then \(v \vec E \predctx t \vec F\).
\item If \(\vec E \predctx \vec F\), then \(\throwto \alpha {\vec E} \predctx \throwto \alpha {\vec F}\).
\item If \(\vec E \predctx \vec F\), then \(\throwto \beta {(\throwto \alpha {\vec E})} \predctx \throwto \alpha {\vec F}\).
\end{enumerate}
\end{definition}

Remark that if we have that \(\cctx {\vec E} {\throwto \alpha q} \pred r\), then \(r\) is not necessarily of the shape \(\cctx {\vec F} {\throwto \alpha {q'}}\) with \(\vec E \pred \vec F\) and \(q \pred q'\) because \(q\) could be a \(\throw\).

\begin{lemma}
\label{lemma:mpred_throw_help}
If \(\cctx {\vec E} {\throwto \alpha {q_1}} \pred r\) and \(q_1 \not\equiv \throwto \gamma s\), 
then there exists a \(q_2\) and \(\vec F\) such that \(r \equiv \cctx {\vec F} {\throwto \alpha {q_2}}\) with \(\vec E \predctx \vec F\) and \(q_1 \pred q_2\).
\end{lemma}

\begin{lemma}
\label{lemma:mpred}
If \(t_1 \pred t_2\), then \(t_2 \pred \MPRED{t_1}\).
\end{lemma}

\begin{proof}
By induction on the derivation of \(t_1 \pred t_2\).
We consider some interesting cases.
\begin{enumerate}
\item Let \(t_1\,r_1 \pred t_2\,r_2\) with \(t_1 \pred t_2\) and \(r_1 \pred r_2\).
	We distinguish the following cases:
	\begin{enumerate}
	\item Let \(t_1 \equiv \lambda x.s_1\) and \(r_1\) a value.
		By distinguishing reductions we have \(t_2 \equiv \lambda x.s_2\) with \(s_1 \pred s_2\).
		Now, \(t_2 \pred \MPRED{t_1}\) and \(s_2 \pred \MPRED{s_1}\) by the \IH{}.
		Furthermore, we have that \(r_2\) is a value by Lemma~\ref{lemma:pred_help}. 
		Therefore, \(t_2\,r_2 \equiv (\lambda x.s_2)\,r_2 \pred \subst {\MPRED {s_1}} x {\MPRED {r_1}} \equiv \MPRED{(t_1\,r_1)}\) by Lemma~\ref{lemma:pred_help}.
	\item Let \(t_1 \equiv \nrecTsymbol\;v_r\;v_s\) and \(r_1 \equiv \nilT\).
		By distinguishing reductions we have \(t_2 \equiv \nrecTsymbol\;r\;s\) and \(r_2 \equiv \nilT\) with \(v_r \pred r\) and \(v_s \pred s\).
		Now, \(r \pred \MPRED{v_r}\) by the \IH{}.
		Therefore, \(t_2\,r_2 \equiv \nrecT r s \nilT \pred \MPRED{v_r} \equiv \MPRED{(\nrecT {v_r} {v_s} \nilT)} \equiv \MPRED{(t_1\,r_1)}\).
	\item Let \(t_1 \equiv \nrecTsymbol\;v_r\;v_s\) and \(r_1 \equiv \consT {v_h} {v_t}\).
		This case is similar to the previous one.
	\item \label{case:mpred_app_throw}
		Let \(t_1 \equiv \cctx {\vec E} {\throwto \beta {q_1}}\) with \(q_1 \not\equiv \throwto \gamma s\).
		By Lemma~\ref{lemma:mpred_throw_help}, we have \(t_2 \equiv \cctx {\vec F} {\throwto \alpha {q_2}}\) with \(\vec E \predctx \vec F\) and \(q_1 \pred q_2\).
		Now we have \(q_2 \pred \MPRED{q_1}\) by the \IH{}.
		Therefore, \(t_2\,r_2 \equiv \cctx {\vec F} {\throwto \alpha {q_2}}\,r_1 \pred \throwto \alpha {\MPRED{q_1}} \equiv \MPRED{(t_1\,r_1)}\).
	\item Let \(r_1 \equiv \cctx {\vec E} {\throwto \beta {q_1}}\) with \(q_1 \not\equiv \throwto \gamma s\) and \(t_1\) a value.
		This proof of this case is similar to the previous one.
	\item For the remaining cases we have \(t_2 \pred \MPRED{t_1}\) and \(r_2 \pred \MPRED{r_1}\) by the \IH{}.
		Therefore, \(t_2\,r_2 \pred \MPRED{t_1}\,\MPRED{r_1} \equiv \MPRED{(t_1\,r_1)}\).
	\end{enumerate}
\item Let \(\catchin \alpha {t_1} \pred \catchin \alpha {t_2}\) with \(t_1 \pred t_2\).
	We distinguish the following cases:
	\begin{enumerate}
	\item Let \(t_1 \equiv \throwto \alpha {q_1}\) with \(q_1 \not\equiv \throwto \gamma s\).
		By distinguishing reductions we obtain that \mbox{\(t_2 \equiv \throwto \alpha {q_2}\)} with \(q_1 \pred q_2\).
		Now we have \(q_2 \pred \MPRED{q_1}\) by the \IH{}.
		Therefore, \(\catchin \alpha {t_2} \equiv \catchin \alpha {\throwto \alpha {q_2}} \pred \catchin \alpha {\MPRED{q_1}} \equiv \MPRED{(\catchin \alpha {t_1})}\).
	\item Let \(t_1 \equiv \throwto \alpha {(\cctx {\vec E} {\throwto \beta {q_1}})}\) with \(q_1 \not\equiv \throwto \gamma s\).
		We have \(t_2 \equiv \cctx {\vec F} {\throwto \beta {q_2}}\) with \(\throwto \alpha {\vec E} \predctx \vec F\) and \(q_1 \pred q_2\) by Lemma~\ref{lemma:mpred_throw_help}.
		Also, \(q_2 \pred \MPRED{q_1}\) by the \IH{}.
		Therefore, \(\catchin \alpha {t_1} \equiv \catchin \alpha {\cctx {\vec F} {\throwto \beta {q_2}}}
			\pred \catchin \alpha {\MPRED{q_1}} \equiv \MPRED{(\catchin \alpha {t_1})}\).
	\item Let \(t_1 \equiv \throwto \beta {v_1}\) with \(\alpha \notin \{\beta\} \cup \FV {v_1}\).
		By distinguishing reductions we obtain that \mbox{\(t_2 \equiv \throwto \beta {v_2}\)} with \(v_1 \pred v_2\).
		Now, \(v_2 \pred \MPRED{v_1}\) by the \IH{}, and \(\alpha \notin \FCV {v_2}\) by Lemma~\ref{lemma:pred_help}.
		So, \(\catchin \alpha {t_2} \equiv \catchin \alpha {\throwto \beta {v_2}}
			\pred \throwto \beta {\MPRED {v_1}} \equiv \MPRED{(\catchin \alpha {t_1})}\).
	\item Let \(t_1\) be a value with \(\alpha \notin \FCV {t_1}\).
		We have \(t_2 \pred \MPRED{t_1}\) by the \IH{}.
		Also, \(t_2\) is a value and \(\alpha \notin \FCV {t_2}\) by Lemma~\ref{lemma:pred_help}.
		Therefore, \(\catchin \alpha {t_2} \pred \MPRED{t_1} \equiv \MPRED{(\catchin \alpha {t_1})}\).
	\item For the remaining cases we have \(t_2 \pred \MPRED{t_1}\) by the \IH{}.
		As a result we have \(\catchin \alpha {t_2} \pred \catchin \alpha {\MPRED{t_1}} \equiv \MPRED{(\catchin \alpha {t_1})}\).
	\end{enumerate}
\item Let \(\cctx {\vec E} {\throwto \alpha {t_1}} \pred \throwto \alpha {t_2}\) with \(t_1 \pred t_2\).
	We distinguish the following cases:
	\begin{enumerate}
	\item Let \(t_1 \equiv \cctx {\vec E} {\throwto \beta {q_1}}\) with \(q_1 \not\equiv \throwto \gamma s\).
		This case is similar to \ref{case:mpred_app_throw}.
	\item For the remaining cases we have \(t_2 \pred \MPRED{t_1}\) by the \IH{}.
		As a result we have \(\throwto \alpha {t_2} \pred \throwto \alpha {\MPRED{t_1}} \equiv \MPRED{(\cctx {\vec E} {\throwto \alpha {t_1}})}\).
	\end{enumerate}
\item Let \(\catchin \alpha {\throwto \alpha {t_1}} \pred \catchin \alpha {t_2}\) with \(t_1 \pred t_2\).
	We have \(t_2 \pred \MPRED{t_1}\) by the \IH{}.
	As a result we have \(\catchin \alpha {t_2} \pred \catchin \alpha {\MPRED{t_1}} \equiv \MPRED{(\catchin \alpha {\throwto \alpha {t_1}})}\).
\item Let \(\catchin \alpha {\throwto \beta {v_1}} \pred \throwto \beta {t_2}\) with \(v_1 \pred t_2\), \(\alpha \notin \{\beta\} \cup \FV {v_1}\).
	We have \(t_2 \pred \MPRED{v_1}\) by the \IH{}.
	Furthermore, \(t_2\) is a value by Lemma~\ref{lemma:pred_help}.
	As a result we have \(\throwto \beta {t_2} \pred \throwto \beta {\MPRED{v_1}} \equiv \MPRED{(\catchin \alpha {\throwto \beta {v_1}})}\).
\item Let \(\catchin \alpha {v_1} \pred t_2\) with \(v_1 \pred t_2\) and \(\alpha\notin\FV {v_1}\).
	We have \(t_2 \pred \MPRED{v_1}\) by the \IH{} and \(t_2\) is a value by Lemma~\ref{lemma:pred_help}.
	Therefore, \(t_2 \pred \MPRED{v_1} \equiv \MPRED{(\catchin \alpha {v_1})}\).
	\qedhere
\end{enumerate}
\end{proof}

\begin{corollary}
\label{corollary:confluence}
If \(t_1 \pred t_2\) and \(t_1 \pred t_3\), then there exists a \(t_4\) such that \(t_2 \pred t_4\) and \(t_3 \pred t_4\).
\end{corollary}

\begin{proof}
Take \(t_4 \defined \MPRED{t_1}\). Now we have \(t_2 \Rightarrow \MPRED{t_1}\) and \(t_3 \Rightarrow \MPRED{t_1}\) by Lemma~\ref{lemma:mpred}.
\end{proof}

\begin{theorem}[Confluence]
\label{theorem:confluence}
If \(t_1 \redd t_2\) and \(t_1 \redd t_3\), then there exists a \(t_4\) such that \(t_2 \redd t_4\) and \(t_3 \redd t_4\).
\end{theorem}

\begin{proof}
By Corollary~\ref{corollary:confluence} and a simple diagram chase (as in~\cite{barendregt1984}), we obtain confluence of \(\pred\).
Now, confluence of \(\red\) is immediate by Lemma~\ref{lemma:red_pred}.
\end{proof}

\section{Strong normalization}
\label{section:sn}
In this section we prove that reduction in \lambdaCatch{} is strongly normalizing.
We use the reducibility method, which is originally due to Tait~\cite{tait1967}.
By this method, instead of proving that a term \(t\) of type \(\rho\) is strongly normalizing, one proves \(t \in \SNint\rho\), where \(\SNint{\sigma\arrow\tau} \defined \{ t \separator \forall s \in \SNint\sigma\ .\ ts \in \SNint\tau \}\).

Although Tait's method does work for the \cbn{} \(\lambdaMu\)-calculus~\cite{parigot1997}, David and Nour~\cite{david2005} have shown that it does not extend to its symmetric variant. 
They proved that the property, if \(r \in \SN\) and \(\subst t x r \in \SNint\sigma\), then \((\lambda x.t)\,r \in \SNint\sigma\), no longer holds due to the reduction \mbox{\(t\,(\mu\alpha.c) \red \mu\alpha.\subst c \alpha {\alpha (t\Box)}\)}.
However, the similar reduction \(t\,(\throwto \alpha r) \red \throwto \alpha r\) in our calculus consumes \(t\) without performing any (structural) substitution in \(r\). 
So, for \lambdaCatch{} this problem does not exist.

It may be possible to prove strong normalization by use of a strictly reduction preserving translation into another system that is already known to be strongly normalizing.
For example, one may try to use the obvious translation into the second-order \cbv{} \lambdaMu{}-calculus where the data type of lists can be defined as 
\(\listTypeT \tau \defined \forall X\,.\; X \arrow (\tau \arrow X \arrow X) \arrow X\).
However, this translation does not preserve the reduction \((\consTsymbol)\).
We are unaware of other systems that are both known to be strongly normalizing, and allow a straightforward strictly reduction preserving translation.

\begin{definition}
\label{definition:sn}
The set of \emph{strongly normalizing terms}, \(\SN\), contains the terms \(t\) for which the length of each reduction sequence starting at \(t\) is bounded.
We use the notation \(\SNbound t\) to denote this bound.
\end{definition}

Due to the addition of lists to \lambdaCatch{}, the interpretation becomes a bit more involved than for the case of \lambdaArrow{}. 
Intuitively, we want our interpretation to ensure that each element of the list \(t \in \SNint{\listTypeT\sigma}\) is contained in \(\SNint\sigma\).

\begin{definition}
Given a set of terms \(S\), the set of terms \(\SNlistint S\) is inductively defined by the following rule.
\[
\AXC{\(\forall v\ w\ .\ \text{if } t \redd \consT v w \text{ then } v \in S \text{ and } w \in \SNlistint S\)}
\UIC{\(t \in \SNlistint S\)}
\normalAlignProof
\DisplayProof
\]
\end{definition}

Notice that the above definition ensures that \(\nilT \in \SNlistint S\) because \(\nilT\) cannot reduce to \(\consT v w\).

\begin{definition}
The interpretation \(\SNint\rho\) of a type \(\rho\) is defined as:
\begin{flalign*}
\SNint{\unitTypeT} &\defined {}
	\SN \\
\SNint{\listTypeT\sigma} &\defined {} 
	\SN \cap \SNlistint{\SNint\sigma} \\
\SNint{\sigma\arrow\tau} &\defined{}
	\{ t \separator \forall s \in \SNint\sigma\ .\ ts \in \SNint\tau \}
\end{flalign*}
\end{definition}

Lemma~\ref{lemma:sn_reducibility_arrow_free} and~\ref{lemma:sn_reducibility} establish an important property: \(\SNint\psi = \SN\) for \arrowFree{} types~\(\psi\).
Since the \(\catch\) operator is restricted to \arrowFree{} types, this means that \(\catchin \alpha r \in \SN\) implies \(\catchin \alpha r \in \SNint\psi\). This property is the key result to prove that \(r \in \SNint\psi\) implies \(\catchin \alpha r \in \SNint\psi\) (Lemma~\ref{lemma:sn_catch}).

The property \(r \in \SNint\sigma\) implies \(\catchin \alpha r \in \SNint\sigma\) does not hold for all types~\(\sigma\).
For example, consider \(t \equiv (\catchin \alpha {\throwto \alpha \omega})\,\omega\) with \(\omega = \lambda x.xx\).
By Corollary~\ref{corollary:sn_throw} we have \(\throwto \alpha \omega \in \SNint{\unitTypeT\arrow\unitTypeT}\) and using the above result we would have had \(t \in \SN\).
This is impossible because \(t \redd \omega\omega \red \omega\omega \red \ldots\)

\begin{definition}
We define the \emph{size} of \(t\), notation \(\sizeof t\), as the number of symbols in \(t\).
For \(t \in \SN\), we define \(\sizeofnf t\) as the size of the normal form of \(t\).
\end{definition}

\begin{lemma}
\label{lemma:sn_reducibility_arrow_free}
If \(\psi\) is \arrowFree{}, then \(\SN \subseteq \SNint\psi\).
\end{lemma}

\begin{proof}
We have to show that for each \(t \in \SN\), we have \(t \in \SNint\psi\).
We proceed by well-founded induction on \(\sizeofnf t\) and a case distinction on the structure of \(\psi\).
The only interesting case is (list), where we have to show that \(t \in \SNlistint{\SNint\psi}\).
So, let \(t \redd \consT v w\) for values \(v\) and \(w\). 
We have \(v \in \SN \subseteq \SNint\psi\) and \(w \in \SNint{\listTypeT\psi}\) by the \IH{} as \(\sizeofnf v < \sizeofnf t\) and \(\sizeofnf w < \sizeofnf t\).
Hence, \(t \in \SNlistint{\SNint\psi}\) as required.
\end{proof}

\begin{lemma}
\label{lemma:sn_reducibility_redd}
If \(t \in \SNint\sigma\) and \(t \redd t'\), then \(t' \in \SNint\sigma\).
\end{lemma}

\begin{proof}
We prove this result by structural induction on \(\sigma\).
\begin{enumerate}
\item[(unit)] Let \(t \in \SNint\unitTypeT = \SN\) and \(t \redd t'\).
	By definition of \(\SN\) we have \(t' \in \SN\).
\item[(list)] Let \(t \in \SNint{\listTypeT\sigma} = \SN \cap \SNlistint{\SNint\sigma}\) and \(t \redd t'\).
	As we have \(t' \in \SN\) by definition of \(\SN\), it remains to prove that \(t' \in \SNlistint{\SNint\sigma}\).
	So, let \(t' \redd \consT v w\) for values \(v\) and \(w\). 
	Now we have \(t \redd t' \redd \consT v w\).
	Therefore, \(v \in \SNint\sigma\) and \(w \in \SNlistint{\SNint\sigma}\) by the assumption that \(t \in \SNlistint{\SNint\sigma}\).
\item[(\(\arrow\))] Let \(t \in \SNint{\sigma\arrow\tau}\) and \(t \redd t'\).
	Since we have to prove that \(t' \in \SNint{\sigma\arrow\tau}\), let \(r \in \SNint\sigma\).
	By assumption we have \(tr\in\SNint\tau\).
	Furthermore we have \(tr \redd t'r\) because \(t \redd t'\).
	Therefore, \(t'r \in \SNint\tau\) by the \IH{}.
	\qedhere
\end{enumerate}
\end{proof}

\begin{definition}
We let \(\vec t\) and \(\vec u\) denote a sequence of terms.
The set \(\SNvec\) contains all sequences of strongly normalizing terms.
\end{definition}

\begin{lemma}
\label{lemma:sn_reducibility}
We have the following results:
\begin{enumerate}
\item \(\SNint\sigma \subseteq \SN\).
\item If \(\vec u \in \SNvec\) then \(x\vec u \in \SNint \sigma\).
\end{enumerate}
\end{lemma}

\begin{proof}
The results are proven simultaneously by structural induction on \(\sigma\).
\begin{enumerate}
\item[(unit)] Both results are immediate. 
\item[(list)] Property (1). \(\SNint{\listTypeT\sigma} = \SN \cap \SNlistint{\SNint\sigma} \subseteq \SN\).

	Property (2). Let \(\vec u \in \SNvec\).
	We have to show that \(x \vec u \in \SNint{\listTypeT\sigma} = \SN \cap \SNlistint{\SNint\sigma}\).
	Since it is immediate that \(x \vec u \in \SN\), it remains to show that \(x \vec u \in \SNlistint{\SNint\sigma}\).
	However, as reductions \(x\vec u \redd \consT v w\) are impossible, we are done.
\item[(\(\arrow\))] Property (1). Let \(t \in \SNint{\sigma\arrow\tau}\).
	We have \(x \in \SNint\sigma\) by the \IH{} of property (2), and therefore \(tx \in \SNint\tau\).
	By the \IH{} of property (1) we have \(\SNint\tau \subseteq\SN\), so \(t \in \SN\).
	
	Property (2). Let \(\vec u \in \SNvec\).
	We have to show that \(x\vec u \in \SNint{\sigma\arrow\tau}\), so let \(r \in \SNint\sigma\).
	By the \IH{} of property (1) we have \(r \in \SN\), and therefore \(x\vec u r \in \SNint\tau\) by the \IH{} of property (2).
	Therefore, \(x\vec u \in \SNint{\sigma\arrow\tau}\) as required.
	\qedhere
\end{enumerate}
\end{proof}

\begin{lemma}
\label{lemma:sn_throw}
If \(r \in \SN\) and \(\vec u \in \SNvec\), then \((\throwto \alpha r)\,\vec u \in \SN\).
\end{lemma}

\begin{proof}
We prove this result by induction on the length of \(\vec u\).
\begin{enumerate}
\item We prove that we have \(\throwto \alpha r \in \SN\) by induction on \(\SNbound r\).
	We proceed by distinguishing the reductions \(\throwto \alpha r \red q\) and show that we have \(q \in \SN\) for each such a \(q\).
	\begin{enumerate}
	\item Let \(\throwto \alpha {(\throwto \beta t)} \red \throwto \beta t\).
		The result holds by assumption.
	\item Let \(\throwto \alpha r \red \throwto \alpha {r'}\) with \(r \red r'\).
		The result follows from the \IH{}.
	\end{enumerate}
\item We prove that we have \((\throwto \alpha r)\,t\,\vec u \in \SN\) by induction on \(\SNbound t + \SNbound{(\throwto \alpha r)\,\vec u}\).
	It is easy to verify that \(q \in \SN\) for all reductions \((\throwto \alpha r)\,t\,\vec u \red q\).
	\qedhere
\end{enumerate}
\end{proof}

\begin{corollary}
\label{corollary:sn_throw}
If \(r \in \SN\) and \(\vec u \in \SNvec\), then \((\throwto \alpha r)\,\vec u \in \SNint\sigma\).
\end{corollary}

\begin{proof}
We prove this result by structural induction on \(\sigma\).
\begin{enumerate}
\item[(unit)] This case is a direct consequence of Lemma~\ref{lemma:sn_throw}.
\item[(list)] We have to show that \((\throwto \alpha r)\,\vec u \in \SNint{\listTypeT\sigma} = \SN \cap \SNlistint{\SNint\sigma}\).
	As we have \((\throwto \alpha r)\,\vec u \in \SN\) by Lemma~\ref{lemma:sn_throw}, it remains to show that \((\throwto \alpha r)\,\vec u \in \SNlistint{\SNint\sigma}\).
	So, let \((\throwto \alpha r)\,\vec u \redd \consT v w\) for values \(v\) and \(w\).
	By distinguishing reductions we see that this reduction is impossible.
\item[(\(\arrow\))] This case follows directly from the \IH{} and Lemma~\ref{lemma:sn_reducibility}.
	\qedhere
\end{enumerate}
\end{proof}

It would be convenient if we could prove \(t \in \SNint\sigma\) by showing that for all reductions \(t \red t'\) we have \(t' \in \SNint\sigma\).
Unfortunately, this result does not hold in general.
For example, whereas the term \(\consT \omega \nilT\) is in normal form, we do not have \(\consT \omega \nilT \in \SNint{\listTypeT{\unitT \arrow \unitT}}\).
Similarly to Girard \etal{}~\cite{girard1989}, we restrict ourselves to the terms \(t\) that are \emph{neutral}.

\begin{definition}
A term is \emph{neutral} if it is not of the shape \(\lambda x.r\), \(\nrecTsymbol\;v_r\;v_s\), or \(\consT v w\).
\end{definition}

\begin{lemma}
\label{lemma:sn_neutral}
If \(t\) is neutral, and for all terms \(t'\) with \(t \red t'\) we have \(t' \in \SNint\sigma\), then \(t \in \SNint \sigma\).
\end{lemma}

\begin{proof}
The results is proven by structural induction on \(\sigma\).
\begin{enumerate}
\item[(unit)] The result is immediate. 
\item[(list)] Let \(t\) be a neutral term such that for all terms \(t'\) with \(t \red t'\) we have \(t' \in \SNint{\listTypeT\sigma}\).
	We have to prove that \(t \in \SNint{\listTypeT\sigma} = \SN \cap \SNlistint{\SNint\sigma}\).
	By Lemma~\ref{lemma:sn_reducibility} we have \(\SNint{\listTypeT\sigma} \subseteq \SN\), and therefore \(t \in \SN\) as \(t' \in \SN\) for each \(t'\) with \(t \red t'\) by assumption.
	It remains to show that \(t \in \SNlistint{\SNint\sigma}\), so let \(t \redd \consT v w\) for values \(v\) and \(w\).
	Since \(t\) is neutral, there should be a term \(t'\) such that \(t \red t' \redd \consT v w\).
	For such a term \(t'\) we have \(t' \in \SNint{\listTypeT\sigma}\) by assumption, hence \(v \in \SNint\sigma\) and \(w \in \SNlistint{\SNint\sigma}\).
	Therefore, \(t \in \SNlistint{\SNint\sigma}\) as required.
\item[(\(\arrow\))] Let \(t\) be a neutral term such that for all terms \(t'\) with \(t \red t'\) we have \(t' \in \SNint{\sigma\arrow\tau}\).
	We have to prove that \(t \in \SNint{\sigma\arrow\tau}\), so let \(r \in \SNint\sigma\).
	By the \IH{} it is sufficient to show that if \(tr \red q\) then \(q \in \SNint\tau\).
	By Lemma~\ref{lemma:sn_reducibility} we have \(r \in \SN\), so we proceed by induction on \(\SNbound r\). 
	We distinguish the following reductions.
	\begin{enumerate}
	\item Let \(tr \red t'r\) with \(t \red t'\). 
		Now we have \(t' \in \SNint{\sigma\arrow\tau}\) by assumption. 
		Hence, \(t'r \in \SNint\tau\) by definition, so we are done.
	\item Let \(tr \red tr'\) with \(r \red r'\). 
		The result follows from the \IH{}.
	\item Let \((\throwto \alpha s)\;r \red \throwto \alpha s\).
			By Lemma~\ref{lemma:sn_reducibility} we have \(\SNint{\sigma\arrow\tau} \subseteq \SN\), and therefore \(\throwto \alpha s \in \SN\) as \(t' \in \SN\) for each \(t'\) with \(\throwto \alpha s \red t'\) by assumption.
			As a consequence we have \mbox{\(\throwto \alpha s \in \SNint\tau\)} by Corollary~\ref{corollary:sn_throw}.
	\item Let \(v\;(\throwto \alpha s) \red \throwto \alpha s\).
		By assumption we have \(\throwto \alpha s \in \SNint\sigma\), so \(\throwto \alpha s \in \SN\) by Lemma~\ref{lemma:sn_reducibility}.
		Hence, \(\throwto \alpha s \in \SNint\tau\) by Corollary~\ref{corollary:sn_throw}.
	\end{enumerate}
	No other reductions are possible because \(t\) is neutral (so, in particular it cannot be of the shape \(\lambda x.s\) or \(\nrecTsymbol\;v_r\;v_s\)).
	\qedhere
\end{enumerate}
\end{proof}

\begin{lemma}
\label{lemma:sn_lambda}
If \(r \in \SN\) and \(\subst t x r \in \SNint\sigma\), then \((\lambda x.t)\,r \in \SNint\sigma\).
\end{lemma}

\begin{proof}
We prove this result by well-founded induction on \(\SNbound t + \SNbound r\).
By Lemma~\ref{lemma:sn_neutral} it is sufficient to show that for each \(q\) with \((\lambda x.t)\,r \red q\) we have \(q \in \SNint\sigma\).
We consider some interesting reductions.
\begin{enumerate}
\item Let \((\lambda x.t)\,v \red \subst t x v\).
	The result holds by assumption.
\item Let \((\lambda x.t)\, (\throwto \beta r) \red \throwto \beta r\).
	In this case we have \(\throwto \beta r \in \SNint\sigma\) by Corollary~\ref{corollary:sn_throw}.
	\qedhere
\end{enumerate}
\end{proof}

\begin{lemma}
\label{lemma:sn_cons}
If \(t \in \SNint\sigma\) and \(s \in \SNint{\listTypeT\sigma}\), then \(\consT t s \in \SNint{\listTypeT\sigma}\).
\end{lemma}

\begin{proof}
First we have to prove that \(\consT t s \in \SN\).
That means, for each \(q\) with \(\consT t s \red q\) we have \(q \in \SN\).
We prove this result by induction on \(\SNbound t + \SNbound s\).
We consider the following reductions.
\begin{enumerate}
\item Let \(\consT {\throwto \alpha r} s \red (\throwto \alpha r)\;s\).
	Since we have \(\throwto \alpha r \in \SNint\sigma\) and \(s \in \SNint{\listTypeT\sigma}\) by assumption, we obtain that \(r,s \in \SN\) by Lemma~\ref{lemma:sn_reducibility}.
	Therefore, \((\throwto \alpha r)\;s \in \SN\) by Lemma~\ref{lemma:sn_throw}.
\item Let \(\consT v {\throwto \alpha r} \red \throwto \alpha r\).
	Since we have \(\throwto \alpha r \in \SNint{\listTypeT\sigma}\) by assumption, we
	obtain that \(\throwto \alpha r \in \SN\) by Lemma~\ref{lemma:sn_reducibility}.
\end{enumerate}
Secondly, we have to prove that \(\consT t s \in \SNlistint{\SNint\sigma}\).
So, let \(\consT t s \redd \consT v w\) for values \(v\) and \(w\).
By distinguishing reductions we obtain that \(t \redd v\) and \(s \redd w\).
Therefore, we have \(v \in \SNint\sigma\) and \(w \in \SNlistint{\SNint\sigma}\) by Lemma~\ref{lemma:sn_reducibility_redd}.
Hence, \(\consT t s \in \SNlistint{\SNint\sigma}\) as required.
\end{proof}

\begin{lemma}
\label{lemma:sn_catch}
If \(\psi\) is \arrowFree{} and \(r \in \SNint\psi\), then \(\catchin \alpha r \in \SNint\psi\).
\end{lemma}

\begin{proof}
By Lemma~\ref{lemma:sn_reducibility_arrow_free} it is sufficient to prove that \(\catchin \alpha r \in \SN\). 
We prove this result by well-founded induction on the lexicographic order on \(\SNbound r\) and \(\sizeof r\). 
Let \(q\) with \(\catchin \alpha r \red q\). 
It remains to prove prove that \(q \in \SN\).
We consider the following interesting reductions.
\begin{enumerate}
\item Let \(\catchin \alpha {\throwto \alpha r} \red \catchin \alpha r\).
	The result follows from the \IH{} as we have \(\SNbound r \le \SNbound{\throwto \alpha r}\) and \(\sizeof r < \sizeof{\throwto \beta r}\).
\item Let \(\catchin \alpha {\throwto \beta v} \red  \throwto \beta v \).
	The result holds by Lemma~\ref{lemma:sn_reducibility}.
\item Let \(\catchin \alpha v \red v\).
	The result holds by Lemma~\ref{lemma:sn_reducibility}.
	\qedhere
\end{enumerate}
\end{proof}

\begin{lemma}
\label{lemma:sn_nrec}
If \(r \in \SNint\rho\), \(s \in \SNint{\sigma \arrow \listTypeT\sigma \arrow \listTypeT\sigma}\), and \(t \in \SNint{\listTypeT\sigma}\), then \(\lrecT r s t \in \SNint\rho\).
\end{lemma}

\begin{proof}
We prove this result by well-founded induction on \(\SNbound r + \SNbound s + \SNbound t + \sizeofnf t\).
By Lemma~\ref{lemma:sn_neutral} it is sufficient to show that for each \(q\) with \(\lrecT r s t \red q\) we have \(q \in \SNint\rho\).
We consider the following interesting reductions.
\begin{enumerate}
\item Let \(\lrecT {v_r} {v_s} \nilT \red v_r\). The result holds by assumption.
\item Let \(\lrecT {v_r} {v_s} {(\consT {v_h} {v_t})} \red v_s\;v_h\;v_t\;(\lrecT {v_r} {v_s} {v_t})\).
	By the definition of \(\consT {v_h} {v_t} \in \SNint{\listTypeT\sigma}\) we obtain that \(v_h \in \SNint\sigma\) and \(v_t \in \SNint{\listTypeT\sigma}\).
	Therefore, we have \(\lrecT {v_r} {v_s} {v_t} \in \SNint\rho\) by the \IH{} as \(\sizeofnf {v_t} \le \sizeofnf{\consT {v_h} {v_t}}\).
	Now, the result follows from the assumption.
\item Let \(\lrecT {(\throwto \alpha r)} s t \red (\throwto \alpha r)\;s\;t\).
	By assumption and Lemma~\ref{lemma:sn_reducibility} we have \(r,s,t \in \SN\), hence \((\throwto \alpha r)\;s\;t \in \SNint\rho\) by Corollary~\ref{corollary:sn_throw}.
	\qedhere
\end{enumerate}
\end{proof}

\begin{corollary}
\label{corollary:sn_aux}
If \(\typed {x_1:\rho_1, \ldots, x_n:\rho_n} \Delta t \tau\) and \(r_i \in \SNint{\rho_i}\) for all \(1 \le i \le n\), then
\[
	t[x_1:=r_1, \ldots, x_n:=r_n] \in \SNint\tau.
\]
\end{corollary}

\begin{proof}
We prove this result by induction on the derivation of \(\typed \Gamma \Delta t \tau\). 
All cases follow immediately from the results proven in this section.
\end{proof}

\begin{theorem}[Strong normalization]
\label{theorem:sn}
If \(\typed \Gamma \Delta t \rho\), then \(t \in \SN\).
\end{theorem}

\begin{proof}
We have \(x_i \in \SNint{\rho_i}\) for each \(x_i:\rho_i \in \Gamma\) by Lemma~\ref{lemma:sn_reducibility}.
Therefore, \(t \in \SNint\rho\) by Corollary~\ref{corollary:sn_aux} and hence \(t \in \SN\) by Lemma~\ref{lemma:sn_reducibility}.
\end{proof}

\section{Conclusions}
\label{section:conclusions}

In this paper we have defined \lambdaCatch{} and proven that it satisfies the usual meta theoretical properties: subject reduction, progress, confluence, and strong normalization.
These proofs require minor extensions of well-known proof methods.
This section concludes with some remarks on possible extensions.

An obvious extension is to add more simple \datatype{}s, like products, sums, finitely branching trees, \etc{}
We expect our proofs to extend easily to these \datatype{}s.
However, adding more complex \datatype{}s presents some challenges. 
For example, consider the type \(\mathtt{tree}\) of unlabeled trees with infinitary branching nodes, with the constructors \(\mathtt{leaf} : \mathtt{tree}\) and \(\mathtt{node} : (\natTypeT \arrow \mathtt{tree}) \arrow \mathtt{tree}\).
A naive extension of the \arrowFree{} restriction would not forbid \(\catchin \alpha {\mathtt{node}\;(\lambda x\,.\,\throwto \alpha {\mathtt{leaf}})}\) which does not reduce to a value.
It would be interesting to modify the \arrowFree{} restriction to avoid this.

Instead of using a \GoedelsTfull{} style recursor, it would be interesting to consider a system with a pattern match and fixpoint construct.
First of all, this approach is more convenient as \GoedelsTfull{} style recursors only allows recursion on direct subterms.
Secondly, this approach would avoid the need for tricks as in Example~\ref{example:predecessor} to improve efficiency.

Another useful extension is to add second-order types \`a la \SystemFfull.
Doing this in a naive way results in either a loss of subject reduction (if we define type variables to be \arrowFree) or makes using \catch{} and \throw{} for the second-order fragment impossible (if we define type variables not to be \arrowFree).

Instead of using the statically bound control operators \catch{} and \throw{}, it would be interesting to consider their dynamically bound variants.
In a dynamically bound \catch{} and \throw{} mechanism, that is for example used in the programming language \CommonLisp{}, substitution is not capture avoiding for continuation variables.
We do not see problems to use such a mechanism instead.

The further reaching goal of this paper is to define a \lambdacalculus{} with \datatype{}s and control operators that allows program extraction from proofs constructed using classical reasoning. 
In such a calculus one can write specifications of programs, which can be proven using (a restricted form of) classical logic.
Program extraction would then allow to extract a program from such a proof where the classical reasoning steps are extracted to control operators.
Herbelin's \IQCMP{}-calculus~\cite{herbelin2010} could be interesting as it includes first-order constructs.

This goal is particularly useful for obtaining provably correct algorithms where the use of control operators would really pay off (for example if a lot of backtracking is performed). 
See~\cite{caldwell2000} for applications to classical search algorithms. 
The work of Makarov~\cite{makarov2006} may also be useful here, as it gives ways to optimize program extraction to make it feasible for practical programming.

\paragraph{Acknowledgments.}
I am grateful to Herman Geuvers and James McKinna for many fruitful discussions, and to the anonymous referees for providing several helpful suggestions.
I thank Freek Wiedijk for feedback on a draft version of this paper.
This work is financed by the Netherlands Organisation for Scientific Research (NWO).

\bibliographystyle{alphaabbr}
\bibliography{bibliography}

\end{document}